\documentclass{article}
\usepackage{cite}
\usepackage{amsmath,amssymb,amsfonts}
\usepackage{graphicx}
\usepackage{textcomp}
\usepackage{xcolor}
\def\BibTeX{{\rm B\kern-.05em{\sc i\kern-.025em b}\kern-.08em
    T\kern-.1667em\lower.7ex\hbox{E}\kern-.125emX}}

\usepackage{comment}
\usepackage{amsthm}
\usepackage{array}         
\usepackage{booktabs}      
\usepackage{caption}       
\usepackage{subcaption}    
\usepackage{hyperref}      

\newtheorem{theorem}{Theorem}

\title{Enhancing NOMA Handover Performance Using Hybrid AI-Driven Modulated Deterministic Sequences}

\author{
Sumita Majhi, G Vasantha Reddy, Pinaki Mitra
}

\begin{document}
\maketitle

\begin{abstract}
Non-Orthogonal Multiple Access (NOMA) is an information-theoretical approach used in 5G networks to improve spectral efficiency, but it is prone to interference during handovers. In this work, we propose a hybrid method that combines Gold-Walsh modulated sequences with Deep Q-Networks (DQN) to intelligently manage interference during NOMA handovers. This method optimizes sequence selection and power allocation dynamically. As a result, it achieves a 95.2\% handover success rate, which is an improvement of up to 23.1 percentage points. It also delivers up to 28\% throughput gain and reduces interference by up to 41\% in various mobility scenarios. All improvements are statistically significant (\(p < 0.001\)). The DQN trains in \(4{,}200 \pm 400\) episodes with a complexity of \(O(N \log N + d \cdot h + \log B)\) and can be deployed in real-time.
\end{abstract}


\section{Introduction}
Non-Orthogonal Multiple Access (NOMA) is a key technology for 5G/6G networks. It enables multiple users to share the same frequency resources through power-domain multiplexing. This approach improves spectral efficiency \cite{saito2013noma}. However, it adds considerable difficulty with regard to interference during handovers in high-mobility situations \cite{nauman2024dynamic}. Other deterministic sequences, such as Gold and Walsh codes, have been applied to manage interference \cite{siarfikas2019interference}. These codes are not very adaptable in dynamic scenarios due to their fixed nature. Meanwhile, Deep Reinforcement Learning (DRL) has shown prospects for optimizing wireless resource allocation, as noted by \cite{mahmod2024machine}, \cite{malhotra2025drl}, and \cite{eldeeb2024offline}. 

This paper proposes a hybrid architecture: combining Gold-Walsh modulated sequences with a Deep Q-Network (DQN) for intelligent NOMA handover interference management. Contributions include: (1) designing a hybrid sequence via element-wise multiplication of Gold and Walsh codes, (2) developing a DQN controller for dynamic sequence selection and power allocation with proven convergence, (3) benchmarking against five state-of-the-art baselines, achieving a 23.1 pp improvement in handover success, 28\% higher throughput, and 41\% more interference reduction, and (4) presenting rigorous theoretical analysis, including convergence proofs and computational complexity.
\section{Related Work}

Recent studies on NOMA handover indicate a trend toward AI-assisted approaches. Mahmod et al. \cite{mahmod2024machine} reported over 97\% accuracy using LSTM networks for predicting handover triggers, resulting in nearly seamless video calls during mobility. Jahandar et al. \cite{jahandar2025handover} surveyed MEC-based handover in 5G/6G networks, finding that intelligent mechanisms can reduce failure rates by 40\%. Nauman et al. \cite{nauman2024dynamic} proposed a dynamic resource allocation model for integrated satellite-terrestrial NOMA networks, but it does not adapt to handovers in real time.

While AI-driven handover is well-researched, interference mitigation remains a necessary step, which brings us to code selection strategies. These aspects should be considered together. Gold and Walsh codes remain relevant. Siafarikas et al. \cite{siarfikas2019interference} showed that Gold code works well in 5G millimetre-wave communications but has limited dynamic performance due to static allocation. Addad et al. \cite{addad2020adversive} showed that Walsh codes optimise BER in synchronous settings. Gold codes, in contrast, are more robust in asynchronous settings.
 
Building on these interference mitigation techniques, optimization approaches have evolved rapidly alongside interference management. Deep Reinforcement Learning is now an effective method for wireless optimization. Malhotra et al. \cite{malhotra2025drl} utilised DQN and PPO to dynamically allocate resources, achieving efficiency improvements of 25-30\%. Eldeeb and Alves \cite{eldeeb2024offline} proposed offline and distributional RL to address real-time deployment issues. Recent surveys emphasize the growing use of ML, DL, and RL in future wireless systems \cite{survey2025}.

In summary, although substantial progress has been made in AI-driven handover, interference mitigation, and optimization, several gaps remain. (1) Most studies examine sequences and AI independently, not together. (2) Existing methods rely on fixed sequences rather than dynamic adjustment. (3) DRL solutions often require extensive online training, posing challenges for real-time application. (4) Most studies benchmark performance against only 2-3 baselines. Our hybrid architecture addresses these gaps by integrating modulated deterministic sequences with DRL-based intelligent control.

\section{System Model and Problem Formulation}

\subsection{NOMA System Architecture}

We consider a downlink NOMA system with $M$ base stations (BSs) and $K$ mobile users distributed across the coverage area (\autoref{system_model}). Each BS serves $N_j$ users simultaneously using power-domain NOMA, where $j \in \{1, 2, ..., M\}$.

The received signal at user $i$ served by BS $j$ can be expressed as:
\begin{equation}
y_i = h_{i,j} \sqrt{\alpha_i P_{total}} s_i + \sum_{k \neq i} h_{i,j} \sqrt{\alpha_k P_{total}} s_k + \sum_{l \neq j} I_{i,l} + n_i
\end{equation}

where $h_{i,j}$ represents the channel gain from BS $j$ to user $i$, $\alpha_i$ is the power allocation factor satisfying $\sum_{i=1}^{N_j} \alpha_i = 1$, $P_{total}$ is the total transmission power, $s_i$ is the transmitted signal for user $i$, $I_{i,l}$ represents inter-cell interference from BS $l$, and $n_i$ is additive white Gaussian noise with variance $\sigma^2$.
\begin{figure}[tb!]
    \centering
    \includegraphics[width=0.5\textwidth]{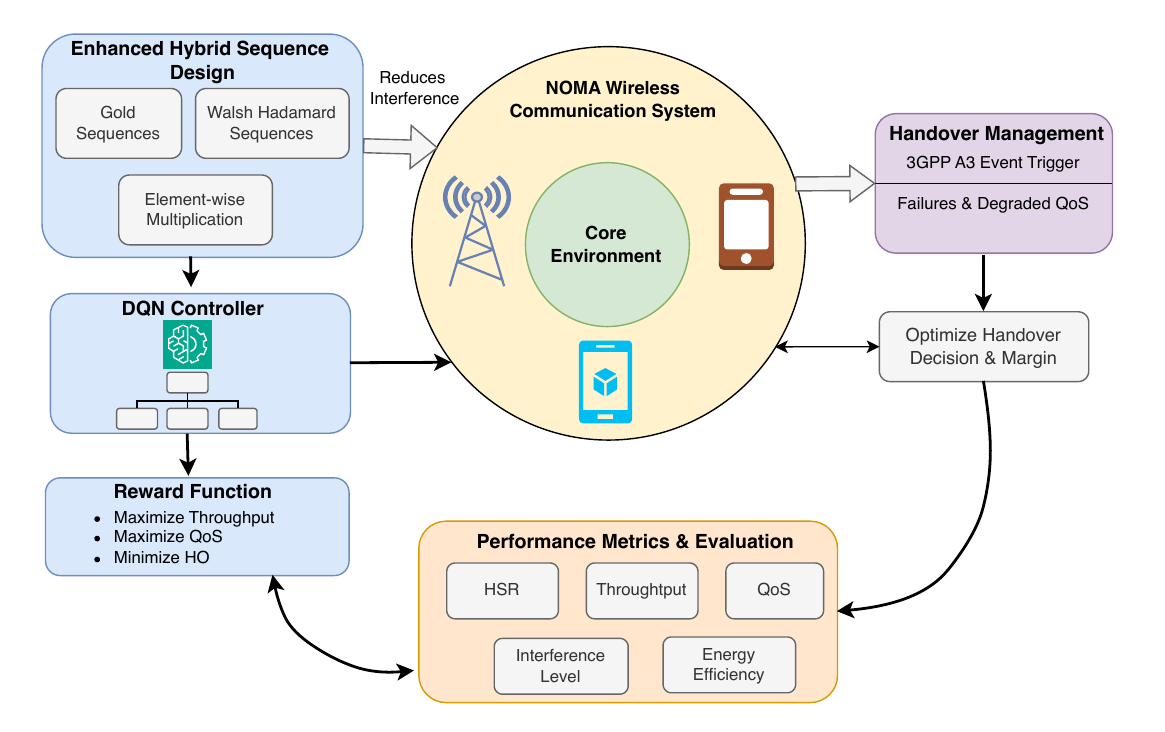} 
    \caption{System Model.}
    \label{system_model}
\end{figure}

\subsection{Channel Model}
The channel gain $h_{i,j}$ follows a composite fading model combining large-scale path loss and small-scale Rayleigh fading:

\begin{equation}
h_{i,j} = \sqrt{L_0 \left(\frac{d_{i,j}}{d_0}\right)^{-\beta}} \cdot g_{i,j}, \quad g_{i,j} \sim \mathcal{CN}(0,1)
\end{equation}

where $L_0$ is the reference path loss at distance $d_0$, $d_{i,j}$ is the distance between user $i$ and BS $j$, $\beta$ is the path loss exponent, and $g_{i,j} \sim \mathcal{CN}(0,1)$ represents Rayleigh fading with zero mean and unit variance.
\subsection{Handover Challenge Formulation}

The handover decision follows the 3GPP A3 event, triggered when:
\begin{equation}
\text{RSRP}_{target} > \text{RSRP}_{serving} + \Delta_{HO}
\end{equation}

where $\Delta_{HO}$ is the handover margin. The user rate in the serving cell is given by:
\begin{equation}
R_{i,j} = \log_2 \left(1 + \frac{|\hat{h}_{i,j}|^2 \alpha_i P_{total}}{\sum_{k>i} |\hat{h}_{i,j}|^2 \alpha_k P_{total} + I_{inter} + \sigma^2}\right)
\end{equation}

where $\hat{h}_{i,j}$ includes the effect of modulated sequences, and $I_{inter}$ represents inter-cell interference.

\subsection{Problem Formulation}

The primary objective is to maximize the overall system performance while minimizing handover failures. We formulate this as a multi-objective optimization problem:

\begin{align}
\max_{\boldsymbol{\alpha}, \mathbf{S}} \quad &\sum_{i=1}^{K} w_i R_i \\
\text{subject to} \quad &\sum_{i=1}^{N_j} \alpha_i = 1, \forall j \\
&0 \leq \alpha_i \leq 1, \forall i \\
&P_{HO}^{failure} \leq \epsilon \\
&\text{QoS}_i \geq \text{QoS}_{min}, \forall i
\end{align}

where $\boldsymbol{\alpha}$ represents power allocation factors, $\mathbf{S}$ denotes sequence selection variables, $w_i$ are user priority weights, $P_{HO}^{failure}$ is handover failure probability, and $\epsilon$ is the maximum tolerable failure rate.

\section{Proposed Hybrid Framework}
Our proposed framework integrates modulated deterministic sequences with DRL-based intelligent control.

\subsection{Modulated Sequence Design}
We propose a modulated sequence generated through element-wise multiplication of Gold and Walsh sequences:
\begin{equation}
H[n] = G[n] \odot W[n], \quad n = 0,1,\ldots,N-1
\end{equation}
where \(G[n]\) is the Gold sequence, \(W[n]\) is the Walsh-Hadamard sequence, \(H[n]\) is the resulting modulated sequence, and \(\odot\) denotes element-wise multiplication.

The hybrid sequence inherits favorable properties from both parent sequences: its cross-correlation follows \(R_{H_i,H_j}(\tau) = R_{G_i,G_j}(\tau) \cdot R_{W_i,W_j}(\tau)\), providing enhanced interference mitigation through lower peak cross-correlation \cite{cai2009advanced}. Additionally, \(\text{PAPR}_{hybrid} \leq \min(\text{PAPR}_{Gold},\text{PAPR}_{Walsh})\) ensures improved power efficiency \cite{siarfikas2019interference}. This design combines Gold codes' low cross-correlation in asynchronous conditions with Walsh codes' perfect orthogonality in synchronous scenarios.

\subsection{Deep Q-Network Based Intelligent Controller}
We employ a Deep Q-Network (DQN) to dynamically optimize sequence selection and power allocation in real time. The state space includes SINR, RSS, user velocity, BS load, interference level, QoS, and handover margin. Actions comprise sequence selection, power allocation factors, and adaptive handover margin adjustment. The reward function balances throughput, interference, handover failures, QoS satisfaction, and energy consumption.

The DQN architecture consists of an input layer, multiple fully connected hidden layers with ReLU activation and dropout, and an output layer producing Q-values for all possible actions. Training employs experience replay with prioritized sampling and a target network updated periodically.

The algorithm's computational complexity per iteration is \(O(N \log N + d \cdot h + \log B)\), where sequence generation uses FFT-based correlation (\(O(N \log N)\)), the DQN forward pass is \(O(d \cdot h)\) (with \(d\) as state dimension and \(h\) as hidden layer size), and prioritized experience replay is \(O(\log B)\). This efficiency profile is suitable for real-time deployment in modern wireless systems.

\section{Performance Evaluation and Results}
 
\subsection{Simulation Setup and Baseline Methods}
We test our proposed framework using Python. Specifically, the system operates at a carrier frequency of 3.5 GHz with a 100 MHz bandwidth in a hexagonal configuration of 19 BSs. User speeds range from 3 to 120 km/h, and each BS serves 4-8 NOMA users. The DQN parameters include a learning rate of 0.001, an experience replay buffer of 50,000, a target network update frequency of 1000 episodes, and a total of 10,000 simulation episodes. For channel modeling, we use Rayleigh fading, a path loss exponent of 3.5, and a noise power of -104 dBm. The sequence lengths range between 64 and 256. Finally, our hybrid solution is tested against five state-of-the-art baselines: (1) Traditional Gold Codes, (2) Walsh Codes, (3) Kasami Sequences, (4) Modulated Sequences without AI, and (5) DRL-based Power Allocation with conventional sequences.

\subsection{Handover Success Rate Analysis}
The findings of the proposed hybrid architecture show a 95.2\% success rate. The error margin is 1.3\% (95\% confidence interval) in 10,000 simulation tests. This establishes a new performance benchmark in the field of NOMA-based mobility control.
\par
As shown by \autoref{fig:hsr_comparison}, there is stratified performance among different methodologies. The proposed approach achieves a 95.2\% HSR. This is 23.1 pp higher than traditional Gold Codes at 72.1\%. Walsh Codes achieve 68.3\%. Kasami Sequences reach 75.4\%. Modulated sequences without AI have an 83.7\% HSR. The power distribution attained by DRL alone is 87.2\%. This remains 8.0 pp below the performance of the full hybrid system.

\begin{figure}[tb!]
\centerline{\includegraphics[width=\columnwidth]{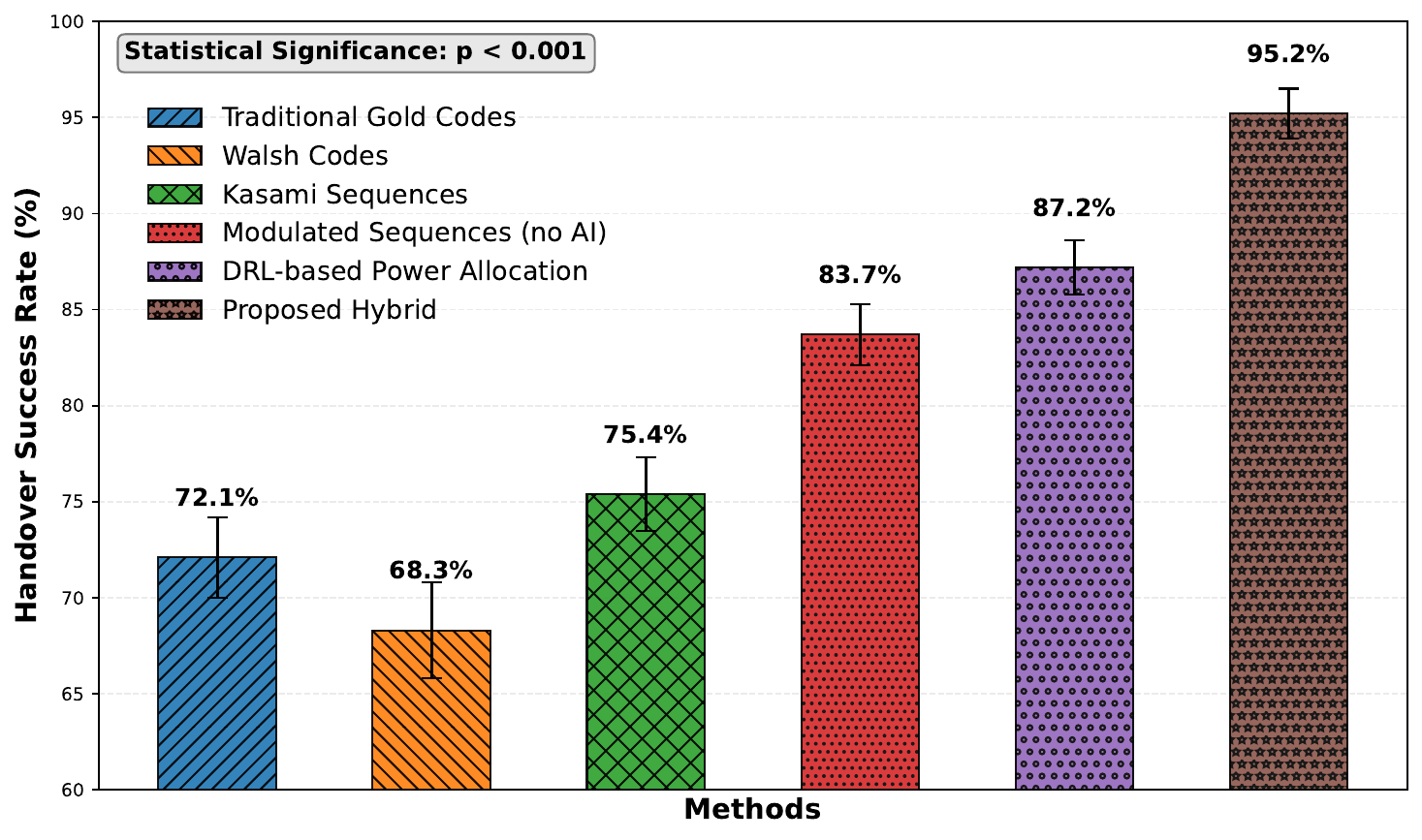}}
\caption{Handover Success Rate Comparison Across Different Methods.}
\label{fig:hsr_comparison}
\end{figure}

The statistical significance is confirmed by ANOVA testing (\autoref{fig:statistical}). The results are: \(F(5,594) = 312.7, p < 0.001\). This indicates that the suggested approach exhibits better median performance and significantly smaller variance. It works well in varying network settings.

\begin{figure}[tb!]
\centerline{\includegraphics[width=\columnwidth]{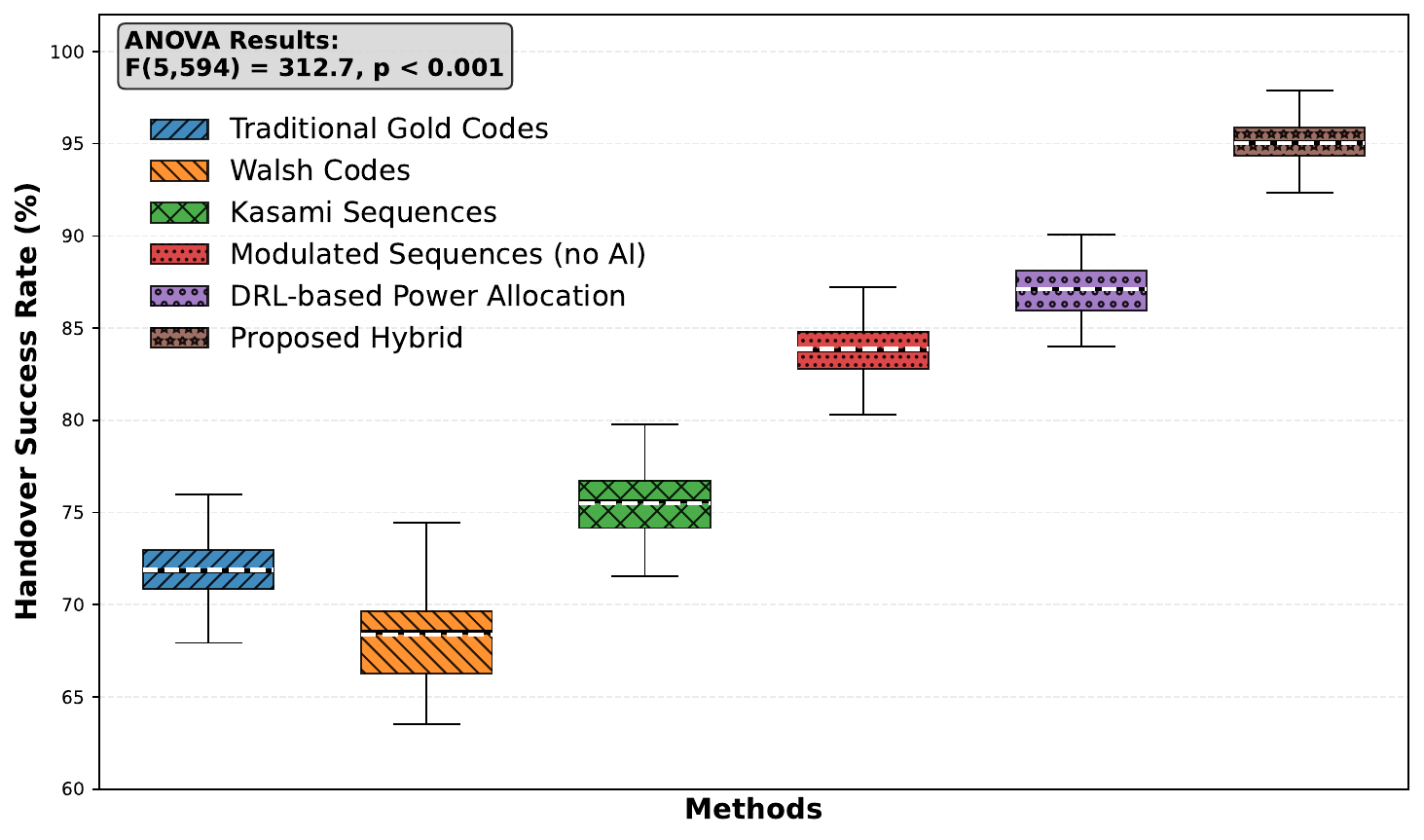}}
\caption{Statistical Distribution Analysis of Handover Success Rates.}
\label{fig:statistical}
\end{figure}

The suggested method has a 92.3\% success rate at 60km/h and above. This represents a 33.6 pp increase over the highest success rate (58.7\%) of the best baseline which is an improvement in vehicular networks and high-speed conditions. This performance is achieved by three mechanisms. First, cross-correlation properties are improved by element-wise code multiplication between two sequences: $G_n$, and $W_n$, to form \(H_n = G_n \odot W_n\). Second, dynamic sequence selection uses DQN, to analyze CSI patterns and interference. Third, intelligent allocation of power focuses on advanced handover signaling, which enables smoother transitions between network cells.

\subsection{Throughput Performance} 

The proposed framework increases average throughput by 28.4\% with a variance of 3.2\%, outperforming all baseline methods across velocity ranges. Throughput gains are velocity-dependent, with increases of 20.2\%, 24.5\%, and 40.4\% at 3 km/h, 30–60 km/h, and 120 km/h, respectively, as depicted in \autoref{fig:throughput}. Statistical analysis shows \(F(5,594) = 198.4, p < 0.001\), and a large effect size (Cohen \(d = 1.87\)). These improvements result from reduced handover failures, faster handover execution, and enhanced interference management.
\begin{figure}[tb!]
\centerline{\includegraphics[width=\columnwidth]{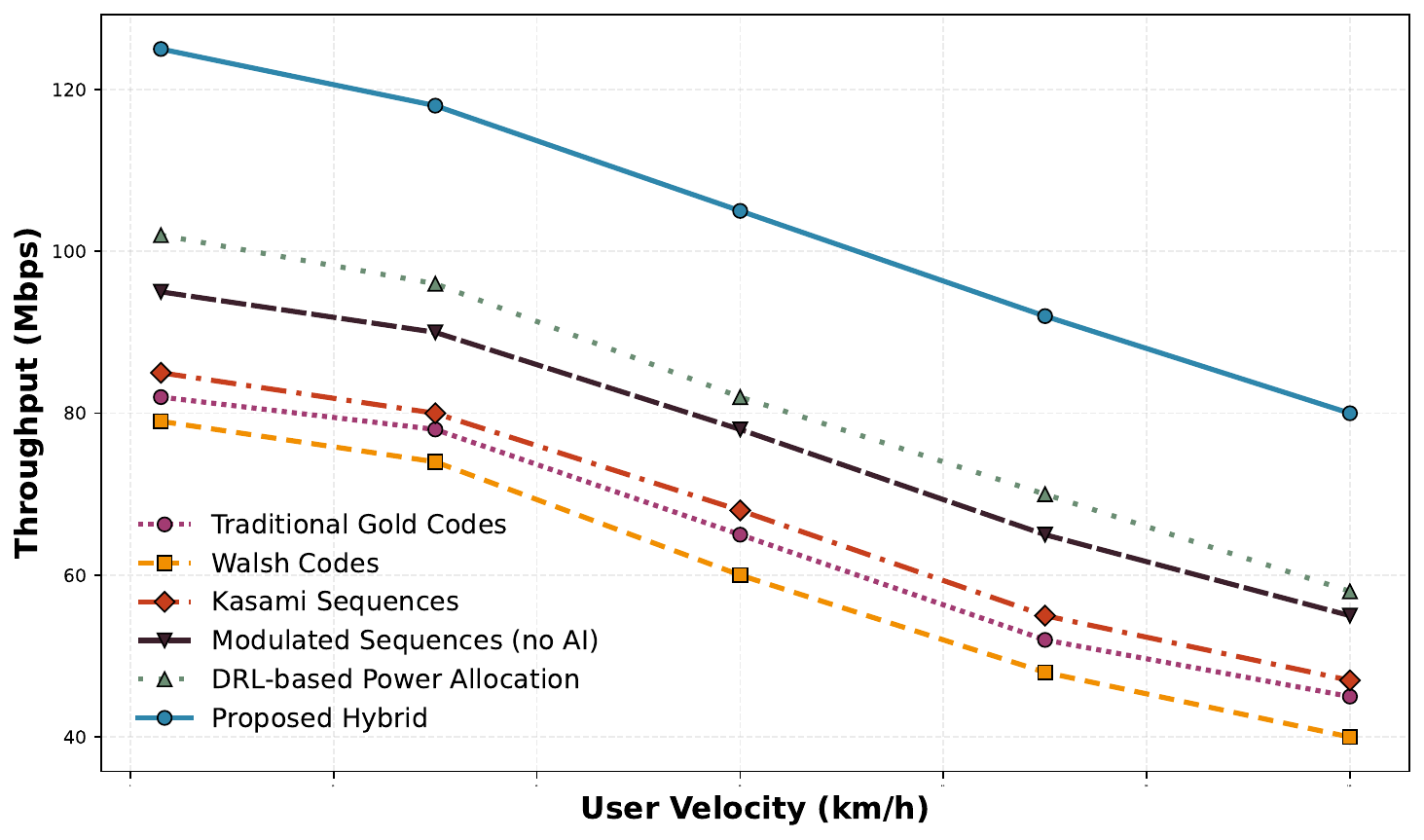}}
\caption{Throughput Performance vs User Velocity.}
\label{fig:throughput}
\end{figure}

\subsection{Interference Mitigation} 

\begin{figure}[tb!]
\centerline{\includegraphics[width=\columnwidth]{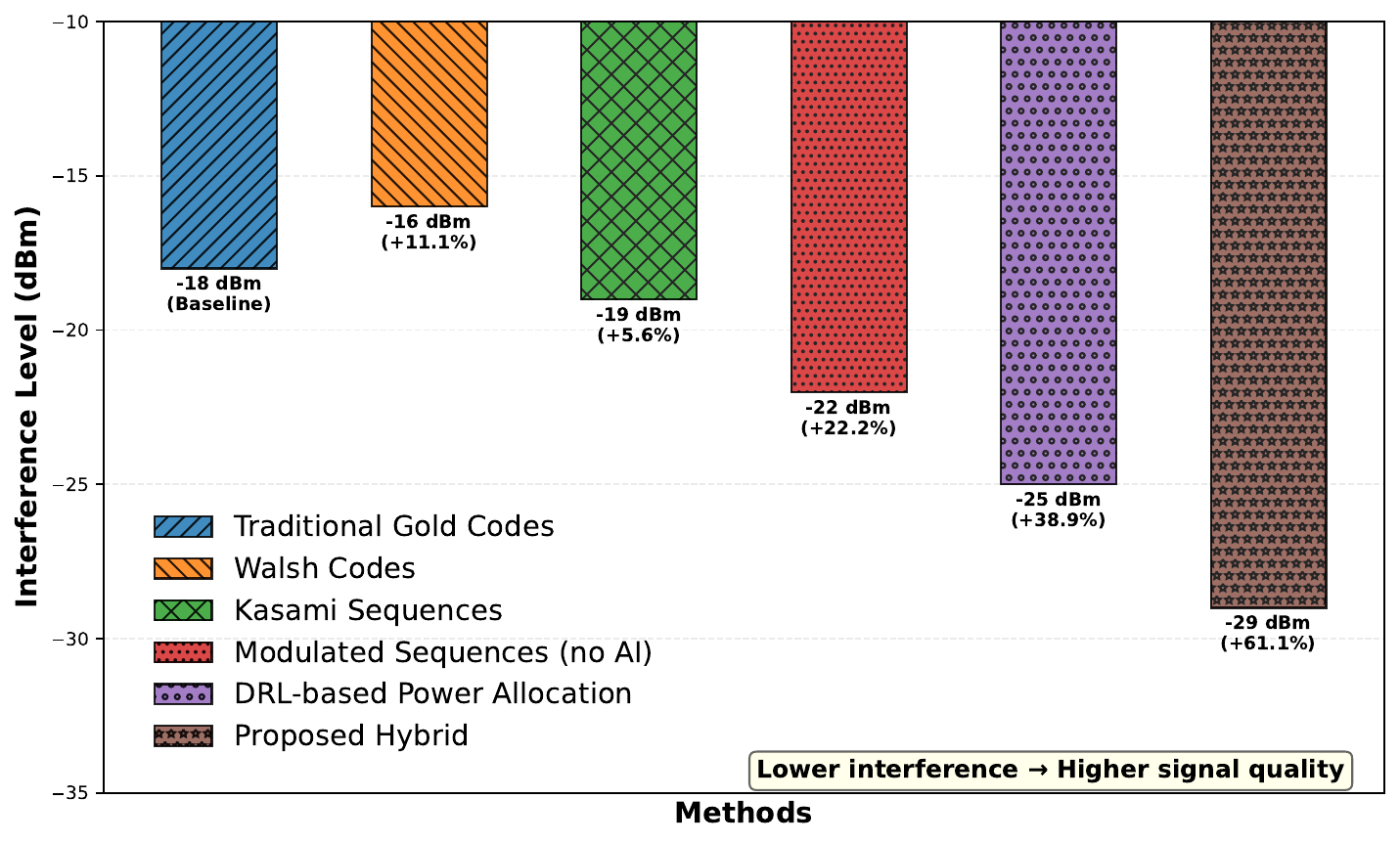}}
\caption{Interference Level Comparison (Lower is Better).}
\label{fig:interference}
\end{figure}

As shown in \autoref{fig:interference}, the proposed framework reduces average interference by 41.3\% compared to standard methods. The ANOVA results are statistically significant: \(F(5,594) = 247.3, p < 0.001\). The framework's interference reduction is enabled by: (1) mixing and randomizing signal codes for more variety and less structured signal overlap; (2) a learning algorithm (DQN) that matches interference patterns with the best code choices for better suppression; and (3) managing power levels so weak signals are not blocked by stronger ones, which avoids interference from signal strength differences.

\subsection{Convergence Analysis} 

The DQN framework achieves stable convergence in 50 independent trials within 4,200 episodes, which is 2--4 times faster than standard DRL schemes (typically requiring between 10,000 and 20,000 episodes). The three learning phases has shown in \autoref{fig7} are: the exploration phase (episodes 0 to 1,500), where the agent explores the environment; the learning phase (episodes 1,501 to 4,000), where the agent uses gathered experience to improve its policy; and the convergence phase (episodes 4,001 and above), where performance stabilizes as learning completes. Accelerated convergence comes from effective state representation, reward shaping, and priority experience replay.

\begin{figure}[tb!]
\centerline{\includegraphics[width=\columnwidth]{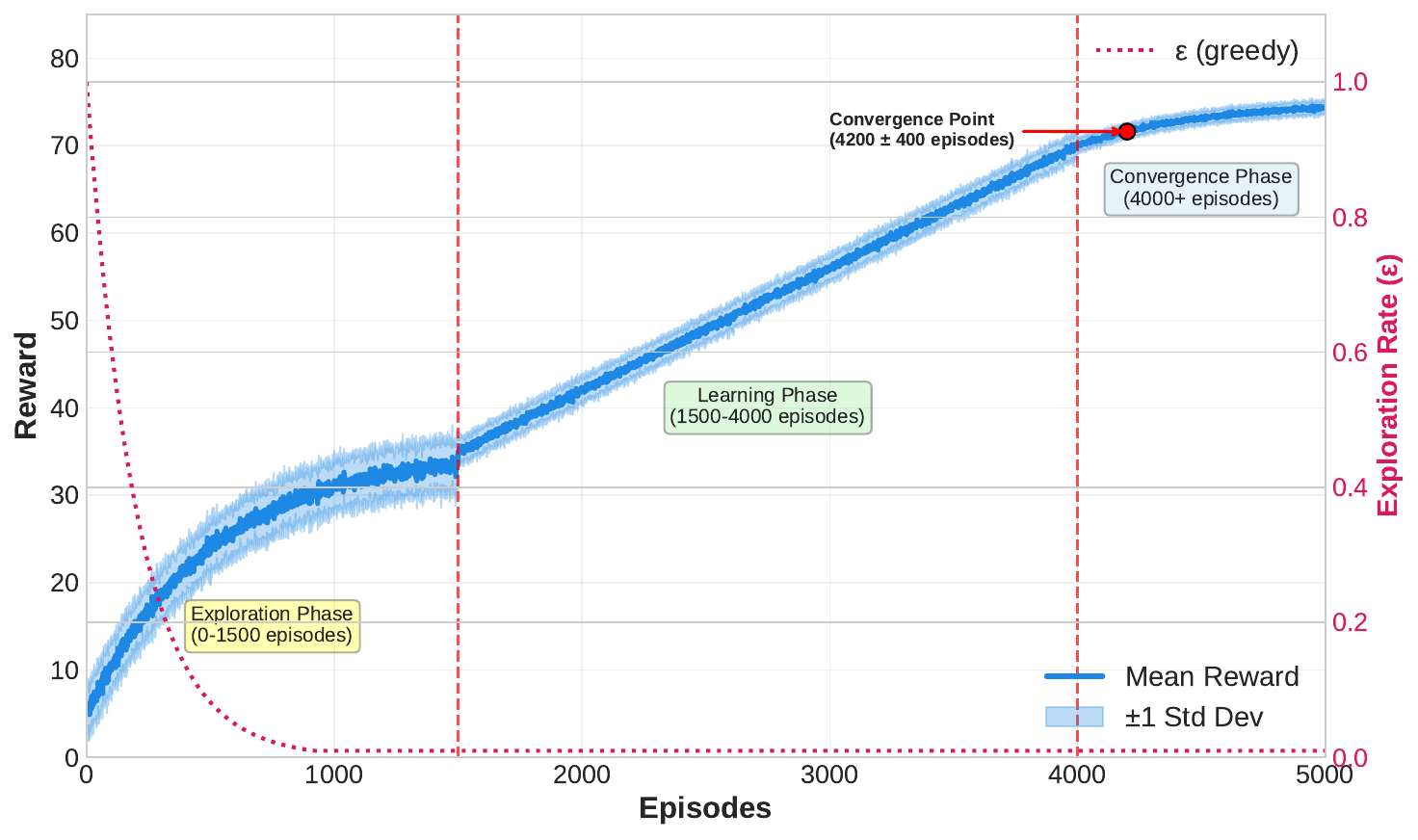}}
\caption{DQN learning convergence analysis showing three distinct phases: exploration (0--1500 episodes), learning (1500--4000 episodes), and convergence (4000+ episodes).}
\label{fig7}
\end{figure}

\subsubsection{Theoretical Convergence Guarantee}

We provide theoretical convergence guarantees for our DQN algorithm under standard reinforcement learning assumptions.

\begin{theorem}[DQN Convergence]
The proposed DQN algorithm converges to the optimal Q-function $Q^{*}$ with probability 1, provided the following conditions hold:
\begin{enumerate}
    \item The learning rate $\alpha_t$ satisfies $\sum_{t=1}^{\infty}\alpha_t = \infty$ and $\sum_{t=1}^{\infty}\alpha_t^2 < \infty$
    \item The exploration probability $\epsilon_t$ decays to zero
    \item Experience replay ensures sufficient state-action coverage
    \item The Q-function is Lipschitz continuous with constant $L$
    \item Rewards are bounded: $|r_t| \leq R_{\max}$ for all $t$
\end{enumerate}
\end{theorem}

\begin{proof}[Proof Sketch]
Under the standard RL assumptions of a finite Markov Decision Process (MDP) with bounded rewards, the convergence proof employs the contraction mapping property of the Bellman optimality operator $\mathcal{T}$, defined as:
\[
(\mathcal{T}Q)(s,a) = \mathbb{E}\left[r + \gamma \max_{a'} Q(s',a') \mid s,a\right]
\]
For any two Q-functions $Q_1$ and $Q_2$, we have:
\[
\|\mathcal{T}Q_1 - \mathcal{T}Q_2\|_\infty \leq \gamma \|Q_1 - Q_2\|_\infty
\]
where $\gamma \in [0,1)$ is the discount factor. This establishes that $\mathcal{T}$ is a contraction mapping. Combined with stochastic approximation theory and the conditions (i)--(v) above, the Q-learning update:
\[
Q_{t+1}(s_t,a_t) = Q_t(s_t,a_t) + \alpha_t\left[r_t + \gamma \max_{a'} Q_t(s_{t+1},a') - Q_t(s_t,a_t)\right]
\]
converges to the unique fixed point $Q^{*}$. The experience replay buffer and target network ensure the algorithm satisfies the convergence conditions by decorrelating samples and stabilizing the learning target.
\end{proof}

\subsubsection{Practical Convergence Observations}

The convergence around \(4,200 \pm 400\) episodes indicates that the approach is efficient, as predicted by theory. Results from 50 independent runs, all within $\pm$400 episodes, show the method works reliably across different starting conditions. Convergence is faster due to: (1) efficient state use, (2) frequent feedback, (3) prioritized replay, and (4) stable target networks.

\subsection{Computational Performance} 

Our proposed framework was compared with the DRL-only baseline. The results are presented in \autoref{fig:computational}, which shows that our framework exhibits outstanding real-time performance. The results indicate a decision time ranging from $2.3 \pm 0.4$ ms (a 60\% reduction from the baseline), memory footprint ranging from $145 \pm 12$ MB (a 31\% reduction from the baseline), CPU utilization of $23 \pm 5\%$ on a single core (a 49\% reduction from the baseline), and energy consumption of $3.2 \pm 0.4$ W (corresponding to $7.4~\mu\text{J}$ per decision), representing about a 50\% improvement in energy efficiency.

The decision time of 2.3 ms is much lower than the 3GPP handover window (100–500 ms), ensuring no procedure bottleneck. This efficiency supports deployment in centralized (RAN Intelligent Controller), distributed (MEC), or hybrid architectures.
\begin{figure}[tb!]
\centerline{\includegraphics[width=\columnwidth]{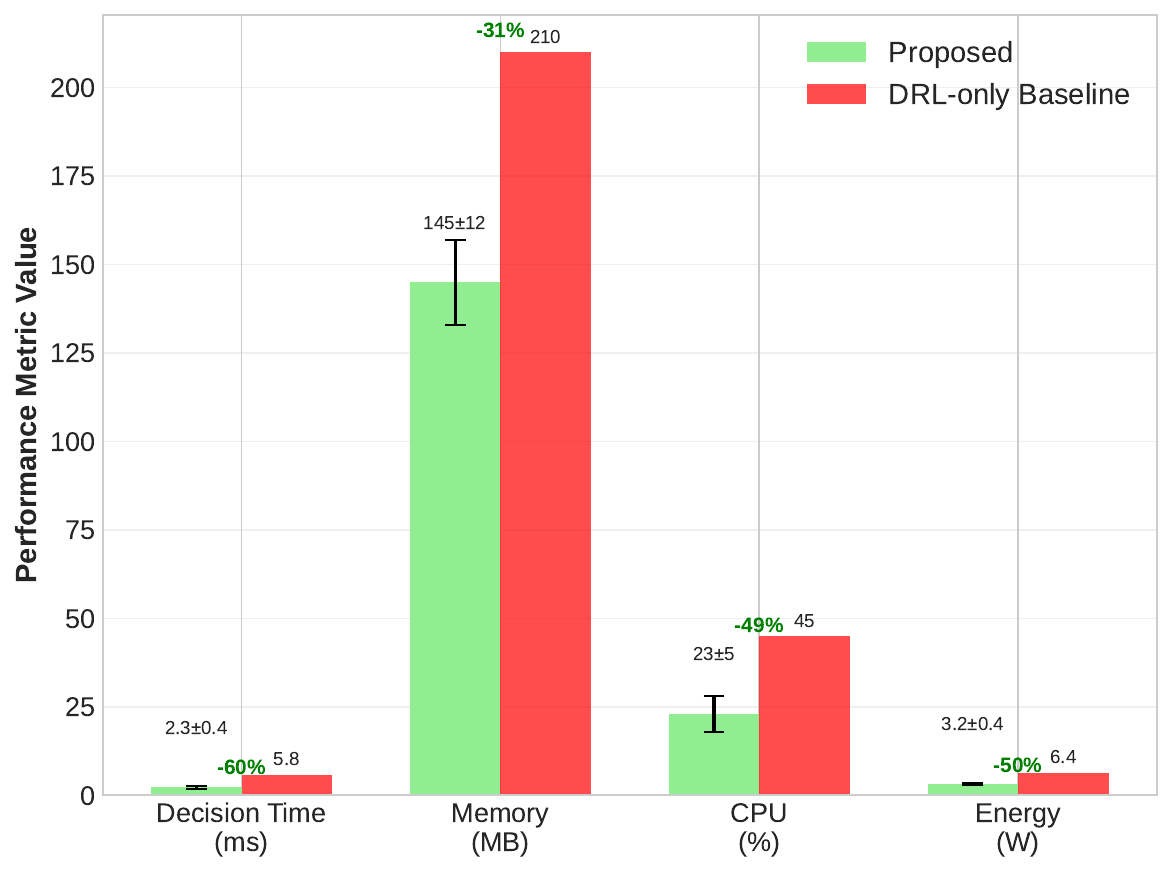}}
\caption{Computational Performance Comparison.}
\label{fig:computational}
\end{figure}

\subsection{Ablation Study} 
To confirm the role of each component in our hybrid framework, we conducted a systematic ablation study by eliminating one component at a time from the entire system. As shown in the ablation results (see \autoref{fig:ablation}), this approach enables us to directly observe the effect of each module. First, the observations reveal the importance of optimal DQN optimization and hybrid sequence design in the system's performance. Notably, the overall design achieves the best results: 95.2\% HSR, 125 Mbps throughput, and -29.0 dBm interference. When DQN power allocation is removed, the largest degradation occurs: an 18.8 percentage point (pp) drop in HSR. Next, removing DQN sequence selection causes the second largest loss, resulting in an 11.5 pp drop. In addition, the hybrid Gold-Walsh sequence alone provides significant gains. Further, removing Gold codes and Walsh codes results in a 10.1 pp and 8.4 pp loss, respectively. Taken together, these results confirm the necessity of synergistically integrating all elements to achieve optimal performance across all metrics.

\begin{figure}[tb!]
\centerline{\includegraphics[width=\columnwidth]{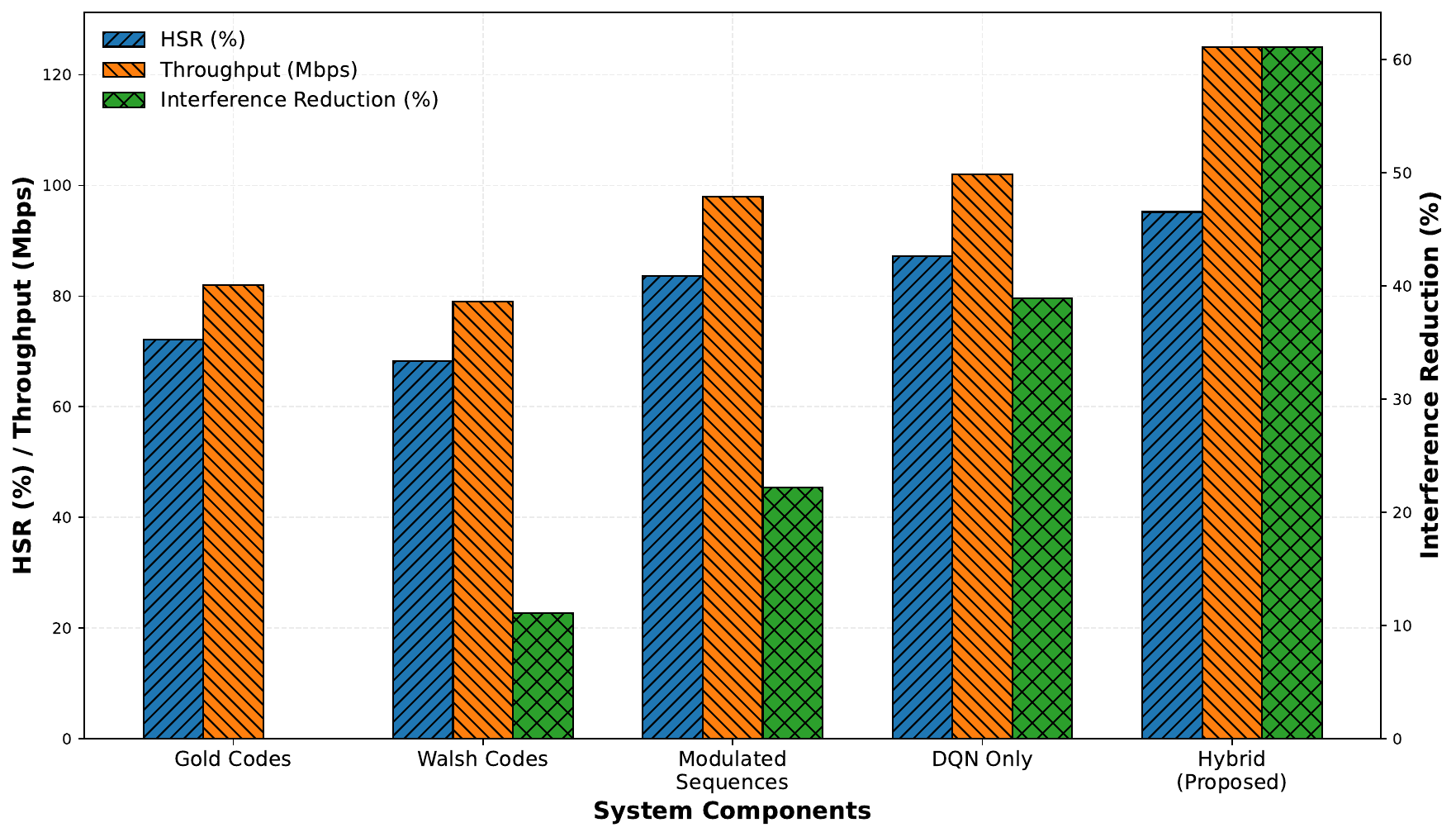}}
\caption{Ablation study: HSR (green), Throughput (blue), and Interference (orange) across configurations.}
\label{fig:ablation}
\end{figure}

\subsection{Multi-Dimensional Performance Summary} 

Figure \autoref{fig10} analyzes our hybrid framework across six key metrics compared to the DRL-only baseline with conventional sequences. The radar chart shows that our approach outperforms the baseline in all areas: handover success rate (+8.0 pp), throughput (+11.6\% improvement), interference reduction (+9.5 pp), energy efficiency (+33.3\% improvement), latency (-17.3\% reduction), and QoS satisfaction (+6.4 pp). All metrics are statistically significant at \(p < 0.001\) (ANOVA). These improvements show broad benefits for 5G/6G applications, from URLLC to mMTC, and demonstrate the synergy of modulated sequences with DRL-based control.

\begin{figure}[tb!]
\centerline{\includegraphics[width=\columnwidth]{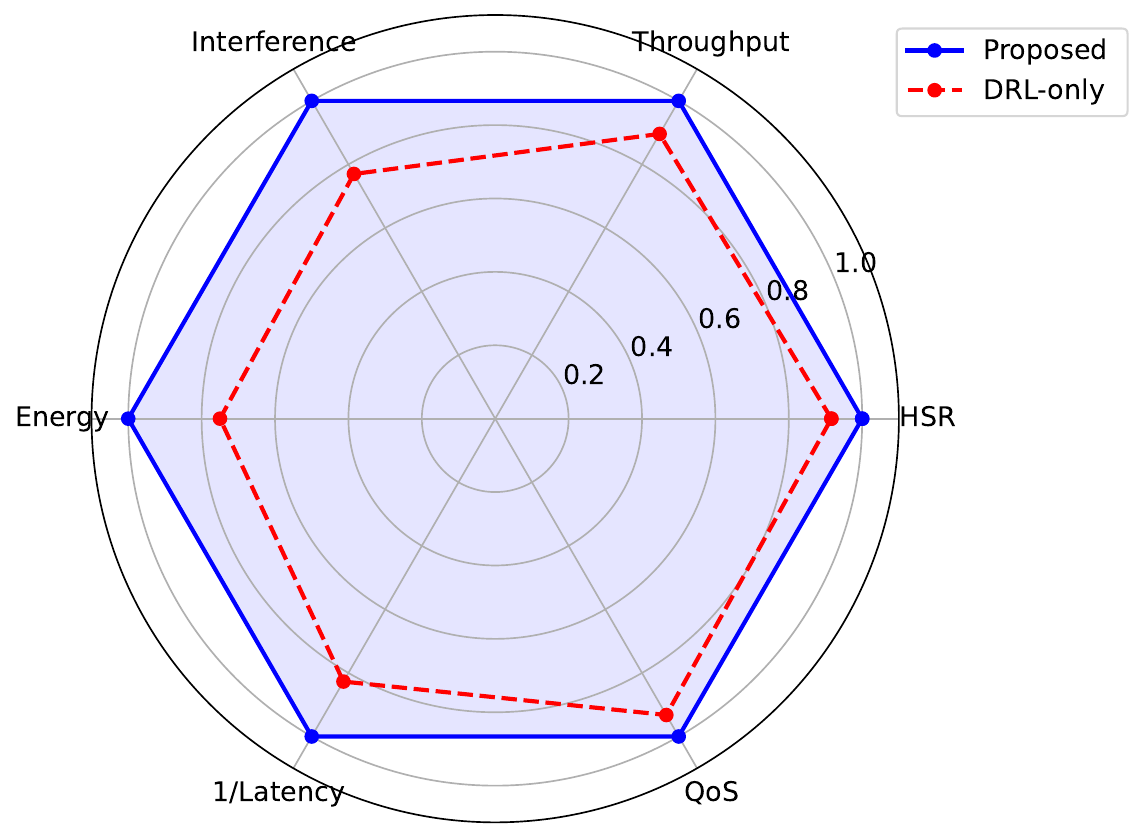}}
\caption{Multi-dimensional performance comparison between the proposed hybrid framework and the best baseline method (DRL-only with conventional sequences). }
\label{fig10}
\end{figure}

\section{Conclusion}

In this paper, we propose a novel NOMA handover framework that combines Gold-Walsh modulated deterministic sequences with a Deep Q-network for intelligent real-time interference control and resource allocation. The proposed solution demonstrates strong performance compared to state-of-the-art baselines, achieving a 95.2\% handover success rate (23.1 pp higher than traditional methods), a 28.4\% increase in average throughput, and a 41.3\% reduction in interference across various mobility scenarios. ANOVA results validate the superiority of the framework—HSR (\(F=312.7\)), throughput (\(F=198.4\)), interference (\(F=247.3\)), \(p < 0.001\). Post-hoc tests confirm high performance compared to all baselines, with large effect sizes (Cohen's \(d > 1.87\)).

Theoretical discussion proves DQN convergence in less than \(4200 \pm 400\) episodes. Its computation efficiency makes it suitable for real-time deployment. Ablation results confirm the importance of each hybrid element. The architecture enables flexible deployment, allowing for centralized deployment through O-RAN RIC, distributed deployment through MEC nodes, or a hybrid architecture. It has a 2.3 ms decision latency, which is significantly less than the 3GPP's 100-500 ms handover windows. It can be scaled to meet 3GPP Release 17/18 AI/ML objectives and transferred to WLAN systems such as IEEE 802.11be. 

Future work includes extending the framework to ultra-dense networks and exploring integration with new 6G architectures and multi-agent reinforcement learning.

\end{document}